\documentclass{article}

\usepackage[affil-it]{authblk}
\usepackage{amsfonts}
\usepackage{amsmath,amsthm,amssymb,dsfont}
\usepackage{enumerate}
\usepackage[english]{babel}
\usepackage{graphicx}	
\usepackage[margin=2.9cm]{geometry}
\usepackage{url}
\usepackage{todonotes}
\usepackage{bbm}

\usepackage{hyperref}
\hypersetup{colorlinks=true,citecolor=blue,linkcolor=blue,filecolor=blue,urlcolor=blue,breaklinks=true}

\usepackage{nicefrac}

\theoremstyle{plain}
\newtheorem{theorem}{Theorem}[section]
\newtheorem{lemma}[theorem]{Lemma}

\newtheorem{corollary}[theorem]{Corollary}
\newtheorem{proposition}[theorem]{Proposition}

\theoremstyle{definition}

\newtheorem{remark}[theorem]{Remark}

\newcommand*{\X}{\mathbb{X}}

\newcommand*{\cN}{\mathcal{N}}
\newcommand*{\cM}{\mathcal{M}}
\newcommand*{\cP}{\mathcal{P}}
\newcommand*{\cQ}{\mathrm{Q}}
\newcommand*{\cR}{\mathcal{R}}

\newcommand*{\cT}{\mathcal{T}}
\newcommand*{\cU}{\mathcal{U}}

\newcommand*{\bPr}{\mathbb{P}}

\newcommand*{\N}{\mathbb{N}}
\newcommand*{\R}{\mathbb{R}}

\newcommand*{\TT}{\mathrm{T}}

\newcommand*{\indic}{\mathds{1}}
\newcommand*{\LL}{\mathrm{L}}
\newcommand*{\St}{\mathrm{S}}

\newcommand*{\id}{\mathrm{id}}
\newcommand*{\poly}{\mathrm{poly}}

\newcommand*{\supp}{\mathrm{supp}}
\newcommand*{\tr}{\mathrm{tr}}
\newcommand*{\ket}[1]{| #1 \rangle}

\newcommand{\proj}[1]{|#1\rangle\!\langle #1|}

\newcommand*{\Pos}{\mathrm{P}}

\newcommand*{\conv}{\mathrm{conv}}
\newcommand*{\TPCP}{\mathrm{TPCP}}

\newcommand*{\MD}{D_{\mathbb{M}}}

\newcommand*{\ci}{\mathrm{i}} 
\newcommand*{\di}{\mathrm{d}} 

\newcommand{\norm}[1]{\left\lVert#1\right\rVert}

\allowdisplaybreaks    

\begin{document}

\title{Strengthened Monotonicity of Relative Entropy\\ via Pinched Petz Recovery Map}

\author[1]{David Sutter}
\author[2]{Marco Tomamichel}
\author[3]{Aram W.~Harrow}

\affil[1]{Institute for Theoretical Physics, ETH Zurich, Switzerland}
\affil[2]{School of Physics, The University of Sydney, Australia}
\affil[3]{Center for Theoretical Physics, MIT, USA}

\date{}

\maketitle
\begin{abstract}
The quantum relative entropy between two states satisfies a monotonicity property meaning that applying the same quantum channel to both states can never increase their relative entropy.  It is known that this inequality is only tight when there is a ``recovery map'' that exactly reverses the effects of the quantum channel on both states. In this paper we strengthen this inequality by showing that the difference of relative entropies is bounded below by the measured relative entropy between the first state and a recovered state from its processed version. The recovery map is a convex combination of rotated Petz recovery maps and perfectly reverses the quantum channel on the second state. As a special case we reproduce recent lower bounds on the conditional mutual information such as the one proved in [Fawzi and Renner, Commun.\ Math.\ Phys., 2015]. Our proof only relies on elementary properties of pinching maps and the operator logarithm.
\end{abstract}

\section{Introduction} \label{sec_intro}
For a quantum state $\rho$ and a non-negative operator $\sigma$ the relative entropy---$D(\rho \| \sigma)\!:=\tr{(\rho (\log \rho - \log \sigma))}$ if the support of $\rho$ is included in the support of $\sigma$, and $+\infty$ otherwise---is a statistical distance measure with numerous applications. Furthermore, several other entropy measures can be derived from the relative entropy, including mutual information, entropy and conditional entropy. 
A fundamental inequality, known as \emph{monotonicity} or \emph{data-processing inequality} of the relative entropy, states that the relative entropy is non-increasing with respect to physical evolutions, i.e., $D(\rho \| \sigma) \geq D(\cN(\rho) \| \cN(\sigma))$ where $\cN$ denotes a quantum channel~\cite{lindblad75,uhlmann77}. It generalizes the celebrated \emph{strong subadditivity} of quantum entropy~\cite{LieRus73_1,LieRus73} which ensures that for any tripartite state $\rho_{ABC}$ the conditional mutual information is non-negative, i.e., $I(A:C|B):=H(AB)+H(BC)-H(ABC)-H(B) \geq 0$, where $H(A):=-\tr(\rho_A \log \rho_A)$ denotes the von Neumann entropy.

Inspired by a question posed in~\cite{WL12}, a series of recent works showed that these fundamental inequalities can be refined in the context of recoverability, i.e., by investigating the question of how good a quantum evolution can be reversed by applying a recovery map. This point of view is motivated by the case where equality holds. It is known that $D(\rho \| \sigma) = D(\cN(\rho) \| \cN(\sigma))$ if and only if there exists a recovery map $\cT$ that simultaneously recovers $\rho$ from $\cN(\rho)$ and $\sigma$ from $\cN(\sigma)$, respectively, i.e., $(\cT\circ \cN)(\rho)=\rho$ and $(\cT \circ \cN)(\sigma)= \sigma$, respectively~\cite{Pet86,Petz88,Pet03}. Furthermore, on the support of $\cT(\sigma)$ the recovery map $\cT$ can be assumed to be of the form
\begin{align} \label{eq_PetzRecMap}
\cT_{\sigma,\cN} \ : \  X \mapsto \sigma^{\frac{1}{2}} \cN^{\dagger} \bigl( \cN(\sigma)^{-\frac{1}{2}} X \, \cN(\sigma)^{-\frac{1}{2}} \bigr) \sigma^{\frac{1}{2}} \ ,
\end{align}
which is referred to as \emph{Petz recovery map} or \emph{transpose map} and $\cN^\dagger$ is the adjoint of $\cN$. A recent result (first proven in~\cite{FR14} and afterwards tightened by~\cite{BHOS14}) states that for any tripartite density operator $\rho_{ABC}$ there exists a recovery map $\cR_{B \to BC}$ such that
\begin{align} \label{eq_FR}
I(A:C|B) &\geq \MD\bigl(\rho_{ABC} \big\| \cR_{B \to BC}(\rho_{AB})\bigr) \\
&\geq - 2 \log F\bigl(\rho_{ABC},\cR_{B \to BC}(\rho_{AB}) \bigr) \ ,
\end{align}
where $\MD(\cdot \| \cdot)$ denotes the \emph{measured relative entropy}. It is defined as the supremum of the relative entropy with measured inputs over all projective measurements\footnote{Without loss of generality these can be assumed to be rank-one projectors.} $\cM = \{M_x\}$, i.e., 
\begin{align} \label{eq_defMeasRelEnt}
&\MD(\rho \| \sigma) := \sup \Big\{ D\bigl( \cM(\rho) \big\| \cM(\sigma) \bigr) :  \cM(\rho) = \sum_{x} \tr(\rho M_x) \proj{x} \text{ with } \sum_{x} M_x = \id \Big\} \ ,
\end{align}
where $\{\ket{x}\}$ is a finite set of orthonormal vectors. This quantity was studied in~\cite{HP91,Hay01}. As observed in~\cite{FBT15} it is not restrictive to consider projective measurements in~\eqref{eq_defMeasRelEnt}, i.e., the supremum over all positive-operator valued measures (POVMs) coincides with the supremum over all projective measurements.\footnote{More precisely, in~\cite{FR14} the second inequality in~\eqref{eq_FR} was proven where in~\cite{BHOS14} the first inequality of~\eqref{eq_FR} was shown for the measured relative entropy defined via an optimization over all possible measurements. Later in~\cite{FBT15} it was shown that without loss of generality we can consider projective measurements only.}

The \emph{fidelity} between two non-negative operators $\rho$ and $\sigma$ is defined as $F(\rho,\sigma):=\norm{\sqrt{\rho}\sqrt{\sigma}}_1$. The second inequality in~\eqref{eq_FR} is a consequence of the monotonicity of quantum R\'enyi divergence in the order parameter~\cite{MLDSFT13} and of the fact that for any two states there exists an optimal measurement that does not increase their fidelity~\cite[Section~3.3]{Fuc96}.
As shown in~\cite{FR14,SFR15}, and independently by a different method in~\cite{Wilde15}, the recovery map satisfying the second inequality of~\eqref{eq_FR} can be assumed to be a rotated Petz map with commuting unitaries.

A closely related and more general statement to the second inequality of~\eqref{eq_FR} in terms of the relative entropy has been proven recently~\cite{Wilde15} (see also the preliminary result in~\cite{BLW14}). It states that for any non-negative operator $\sigma$, any density operator $\rho$ such that $\supp(\rho) \subseteq \supp(\sigma)$, and any quantum channel $\cN$ we have\footnote{We recover the second inequality of~\eqref{eq_FR} by choosing $\rho = \rho_{ABC}$, $\sigma=\rho_{BC}$ and $\cN(\cdot)=\tr_C(\cdot)$.}
\begin{align} \label{eq_Wilde}
D(\rho \| \sigma) - D\bigl( \cN(\rho) \big\| \cN(\sigma) \bigr) \geq -2 \log \Bigl( \sup_{t \in \R} F\bigl( \rho, (\cU_{\sigma,t}\circ \cT_{\sigma,\cN} \circ \cU_{\cN(\sigma),-t} \circ \cN)(\rho) \bigr) \Bigr) \ ,
\end{align}
where $\cU_{\kappa,t}(\cdot):=\kappa^{\ci t} (\cdot) \kappa^{-\ci t}$. We note that this ensures the existence of a recovery map satisfying~\eqref{eq_Wilde} that perfectly reconstructs $\sigma$ from $\cN(\sigma)$.\footnote{This map does not necessarily have the form given in~\eqref{eq_Wilde}.}
Furthermore, the recovery map that fulfills~\eqref{eq_Wilde} has a potential dependence on $\rho$ (hidden in the supremum).

The original proof of~\eqref{eq_FR} given in~\cite{FR14} relies on properties of R\'enyi entropies and de Finetti arguments. Subsequently in~\cite{BHOS14}, a different proof has been presented that is based on state redistribution and also de Finetti type arguments. (As shown recently in~\cite{TB15}, the de Finetti type arguments can be replaced by an argument based on semidefinite programming.) In~\cite{Wilde15}, a proof for~\eqref{eq_Wilde} has been presented that uses the Hadamard three-line theorem.

Let $\LL(A,B)$ be the set of bounded linear operators from $A$ to $B$. A linear map $\cN \in \LL(A,B)$ is called \emph{$n$-tensor-stable positive} for some number $n \in \N$ if the map $\cN^{\otimes n}$ is positive. A linear map $\cN \in \LL(A,B)$ is called \emph{tensor-stable positive} if the map $\cN$ is $n$-tensor stable positive for all $n \in \N$. Completely positive maps are tensor stable positive, but this is not a necessary condition (see~\cite{MRW15} for an overview and characterization of tensor-stable positive maps). We note that it is well-known that the data processing inequality remains valid when considering trace-preserving tensor-stable positive linear maps (see, e.g.~\cite[Theorem~5.5]{hayashi_book}).

\subsection{Main result.}
We prove that for any non-negative operator $\sigma$, any density operator $\rho$ such that $\supp(\rho) \subseteq \supp(\sigma)$, and any trace-preserving tensor-stable positive linear map $\cN$ there exists a recovery map $\cR$ such that $(\cR\circ \cN)(\sigma) = \sigma$ and
\begin{align} \label{eq_overviewOurResult}
D(\rho \| \sigma) - D\bigl(\cN(\rho) \big\| \cN(\sigma) \bigr) 
\geq \MD\bigl(\rho \, \big\| (\cR \circ \cN)(\rho)\bigr) \geq -2 \log  F\bigl(\rho , (\cR \circ \cN)(\rho)\bigr) \ .
\end{align}
We refer to Theorem~\ref{thm_monoMeasRel} for a more precise statement. We note that this is a strengthened version of~\eqref{eq_Wilde} that directly implies~\eqref{eq_FR}. Our proof is conceptually different from all the previous ones. It is based on a particular choice of a recovery map that is constructed out of \emph{pinching maps}. These pinching maps fulfill several desirable properties that offer a simple proof for the new inequality.

 \subsection{Notation.}
For $n \in \N$ we define $[n]:=\{1,2,\ldots,n\}$. 
 In this work, all Hilbert spaces are assumed to be finite-dimensional. Let $\St(A)$ denote the set of density operators on a Hilbert space $A$ and for $\sigma \in \St(A)$ let $\St_{\sigma}(A):=\{\rho \in \St(A): \supp(\rho) \subseteq \supp(\sigma)\}$, where $\supp(\rho)$ is the support of $\rho$. $\Pos(A)$ is the set of non-negative operators on $A$.  
The set of trace-preserving completely positive maps from $A$ to $B$ is denoted by $\TPCP(A,B)$. We use $\indic_{\{\textnormal{statement}\}}$ to denote the indicator of the statement. It is equal to $1$, if the statement is true, and it is equal to $0$, if the statement is false.

\section{Variations on the Petz Recovery Map} \label{sec_variantsRecMaps}

In this section, we introduce two recovery maps that are closely related to the Petz recovery map defined in~\eqref{eq_PetzRecMap}. One of them is based on the concept of \emph{pinching}. The other one is a rotated version of the Petz recovery map. These two recovery maps play a key role in this work.
We further show a close connection between the two recovery maps (see Lemma~\ref{lem_pinchingVsRotPetz}).

\subsection{Pinching recovery map.} 
Let $H \in \LL(A,A)$ be a self-adjoint operator with eigenvalue decomposition $H=\sum_{x} \mu_x P_{x}$, where $x$ ranges between $1$ and $d = \dim(A)$ and $P_{x}$ is the projector onto the eigenspace corresponding to the eigenvalue $\mu_x$.  The \emph{pinching map} for this spectral decomposition of the operator $H$ is given by
 \begin{align}
 \cP_{H} \ : \ K \mapsto \sum_{x}  P_x \,  K \,  P_{x} \ .
 \end{align}
 Such maps are trace-preserving, completely positive, unital, self-adjoint, and can be viewed as dephasing operations that remove off-diagonal blocks of a matrix.
 
Similarly, for any $n \in \mathbb{N}$ we define projectors onto type subspaces on $A^{\otimes n}$ corresponding to the eigendecomposition of $H$. 
Namely, we say that a sequence $\{ x_1, x_2, \ldots, x_n \}$ is of type $\lambda = \big(\lambda_1, \lambda_2, \ldots, \lambda_{d} \big)$ if $\sum_{i=1}^n \indic_{\{x_1=x\}} = \lambda_x$ for all $x$.\footnote{More information about the method of types can be found, e.g., in~\cite{csiszar98}.}
For such a type $\lambda$, we define
\begin{align}
  P_{\lambda} := \sum_{\{ x_1, x_2, \ldots, x_n \} \textrm{ of type } \lambda } \ \bigotimes_{i=1}^n P_{x_i} ,
\end{align}
and the pinching map 
\begin{align}
  \cP_{H,n} \ : \ K \mapsto \sum_{\lambda \in \Lambda_n}  P_{\lambda} \,  K \,  P_{\lambda}
\end{align}
where the sum goes over the set $\Lambda_n$ of all types of $n$ symbols in the set $\{1, 2, \ldots, d\}$.  Clearly $\cP_{H,1} = \cP_H$.
The number of such types is bounded according to~\cite[Lemma~II.1]{csiszar98}, i.e.\
\begin{align}
|\Lambda_n| \leq 
\left(\begin{matrix} n + d - 1 \\ d - 1 \end{matrix}
 \right) \leq \frac{(n+d-1)^{d -1}}{(d -1 )!} = O\bigl(\poly(n)\bigr) \ . \label{eq_csiszar}
\end{align}
 
 The pinching map fulfills four properties that are heavily used in this article: (i) $\cP_{H,n}(X)$ commutes with $H^{\otimes n}$ for any $X \in \Pos(A^{\otimes n})$, (ii) $\cP_{H,n}(H^{\otimes n}) = H^{\otimes n}$, (iii) $\tr(\cP_{H,n}(X) H^{\otimes n})=\tr(XH^{\otimes n})$ for any $X \in \Pos(A^{\otimes n})$, and (iv) it satisfies the following operator inequality (sometimes referred to as \emph{Hayashi's pinching inequality}~\cite{hayashi02}) 
\begin{align} \label{eq_hayashiPinchingIneq}
\cP_{H,n}(X) \geq \frac{1}{|\Lambda_n|} X \quad \textnormal{for all } X \in \Pos(A^{\otimes n})  \ .
\end{align}
More information about pinching maps together with a simple proof of~\eqref{eq_hayashiPinchingIneq} can be found in~\cite[Section~2.6.3]{Marco_book} (see also~\cite[Section~4.4]{carlen_book}). For any $\sigma \in \Pos(A)$, any $\cN \in \TPCP(A,B)$, and $n \in \N$ we define the \emph{pinching recovery map} $\cR^n_{\sigma,\cN}:  \Pos(B^{\otimes n}) \to \Pos(A^{\otimes n})$ by
 \begin{align} \label{eq_recMap}
X_{B^n} \mapsto & \,\, \bigl(\cP_{\sigma,n}\circ (\cT_{\sigma,\cN})^{\otimes n}\circ \cP_{\cN(\sigma),n}\bigr)(X_{B^n}) \nonumber \\
 &\hspace{12mm}=(\sigma^{\frac{1}{2}})^{\otimes n} \cP_{\sigma,n}\Big((\cN^{\dagger})^{\otimes n} \Big[ \big(\cN(\sigma)^{-\frac{1}{2}}\bigr)^{\otimes n} \cP_{\cN(\sigma),n} (X_{B^n}) \big(\cN(\sigma)^{-\frac{1}{2}}\big)^{\otimes n}  \Big]  \Big)  (\sigma^{\frac{1}{2}})^{\otimes n}  \ ,
\end{align}
where we used the property (i) of pinching maps for the equality step.
This map is trace-preserving and completely positive for any $n \in \N$. It is clearly completely positive, as it can be written as a composition of completely positive maps. That it is trace-preserving follows by employing some properties of pinching maps mentioned above. To simplify the notation we consider the case where $n=1$. (The generalization to $n\in \N$ is immediate.)
For any $X_B \in \Pos(B)$ we find
\begin{align}
\tr\left( \sigma^{\frac{1}{2}} \cP_{\sigma} \left(\cN^{\dagger} \bigl( \cN(\sigma)^{-\frac{1}{2}} \cP_{\cN(\sigma)} (X_B) \cN(\sigma)^{-\frac{1}{2}}  \bigr)  \right)  \sigma^{\frac{1}{2}} \right)
&= \tr \Bigl( \sigma \cN^{\dagger} \bigl( \cN(\sigma)^{-\frac{1}{2}} \cP_{\cN(\sigma)} (X_B) \cN(\sigma)^{-\frac{1}{2}}  \bigr)  \Bigr) \\
&= \tr \bigl( \cN(\sigma)  \cN(\sigma)^{-\frac{1}{2}} \cP_{\cN(\sigma)} (X_B) \cN(\sigma)^{-\frac{1}{2}} \bigr) \\
&=\tr\bigl( \cP_{\cN(\sigma)} (X_B) \bigr) \\
&=\tr(X_B) \ ,
\end{align}
where we used properties (i) and (iii) of the pinching map together with the definition of the adjoint channel.

\subsection{Rotated Petz recovery map.}
For any $\sigma \in \Pos(A)$ let $\sigma = \sum_{k \in [d_1]} \lambda_k P_k$ be an eigenvalue decomposition of $\sigma$. Here, $d_1 \leq \dim(A)$, $\{ P_k \}_{k \in [d_1]}$ are mutually orthogonal projectors, and the $\lambda_k$ are all different. Using this and a vector $\vartheta = (\vartheta_1,\ldots,\vartheta_{d_1}) \in [0,2\pi]^{\times d_1}$, we define the unitary $U_{\sigma}^{\vartheta} := \sum_{k \in [d_1]} \exp(\ci \vartheta_k) P_k$.
Furthermore, for any $\cN \in \TPCP(A,B)$ and an eigenvalue decomposition of $\mathcal{N}(\sigma)$, we use the same construction to define $U_{\mathcal{N}(\sigma)}^{\varphi}$. 
  We note that due to the evident fact that $P_j P_k = \indic_{\{j=k\}} P_j$, the unitaries $U_{\sigma}^{\vartheta}$ and $U_{\cN(\sigma)}^{\varphi}$ commute with $\sigma$ and $\cN(\sigma)$, respectively, for any $\vartheta \in [0,2\pi]^{\times d_1}$ and any $\varphi \in [0,2\pi]^{\times d_2}$ with $d_1 \leq \dim(A)$ and $d_2 \leq \dim(B)$.
  
With the help of these unitaries we define a (doubly) \emph{rotated Petz recovery map} $\cT_{\sigma,\cN}^{\varphi,\vartheta}\ :\  \Pos(B) \to \Pos(A)$ by
\begin{align} \label{eq_rotPetzMap}
   &X_B\! \mapsto  U_{\sigma}^{\vartheta} \sigma^{\frac{1}{2}} \cN^{\dagger} \bigl( \cN(\sigma)^{-\frac{1}{2}} U_{\mathcal{N}(\sigma)}^{\varphi} X_{B} U_{\mathcal{N}(\sigma)}^{\varphi\dag} \cN(\sigma)^{-\frac{1}{2}} \bigr)   \sigma^{\frac{1}{2}} U_{\sigma}^{\vartheta\dag} \ .
\end{align}
 For any $\sigma \in \Pos(A)$ and $\cN \in \TPCP(A,B)$ we denote the convex hull of rotated Petz recovery maps by 
\begin{align} \label{eq_convexHullRotatedPetz}
\TT_{\sigma,\cN}:=\conv\Bigl(\cT_{\sigma,\cN}^{\varphi,\vartheta} \ : \ \vartheta \in [0,2\pi]^{\times d_1}, \varphi \in[0,2\pi]^{\times d_2}  \Bigr) \ .
\end{align}
The recovery map $\cT_{\sigma,\cN}^{\varphi,\vartheta}$ is trace-preserving and completely positive. It is clearly completely positive, as it can be written as a composition of completely positive maps. Furthermore, we find for any $X_B \in \Pos(B)$
\begin{align}
& \tr \Bigl( U_{\sigma}^{\vartheta} \sigma^{\frac{1}{2}}  \cN^{\dagger} \bigl( \cN(\sigma)^{-\frac{1}{2}} U_{\mathcal{N}(\sigma)}^{\varphi} X_{B} U_{\mathcal{N}(\sigma)}^{\varphi\dag} \cN(\sigma)^{-\frac{1}{2}} \bigr)  \sigma^{\frac{1}{2}} U_{\sigma}^{\vartheta\dag}   \Bigr) \nonumber\\
&\hspace{50mm} = \tr\Bigl(\sigma \cN^{\dagger} \bigl( \cN(\sigma)^{-\frac{1}{2}} U_{\mathcal{N}(\sigma)}^{\varphi} X_{B} U_{\mathcal{N}(\sigma)}^{\varphi\dag} \cN(\sigma)^{-\frac{1}{2}} \bigr) \Bigr) \\
&\hspace{50mm} =\tr\bigl(\cN(\sigma) \cN(\sigma)^{-\frac{1}{2}}  U_{\mathcal{N}(\sigma)}^{\varphi} X_{B} U_{\mathcal{N}(\sigma)}^{\varphi\dag} \cN(\sigma)^{-\frac{1}{2}}  \bigr) \\
&\hspace{50mm}= \tr(X_B) \ ,
\end{align}
which shows that $\cT_{\sigma,\cN}^{\varphi,\vartheta}$ is trace-preserving.

\subsection{Connection between pinching and rotated Petz recovery map.}
The following lemma shows that the pinching recovery map (defined in~\eqref{eq_recMap}) can be written as a convex combination of tensor products of rotated Petz recovery maps (defined in~\eqref{eq_rotPetzMap}). This connection will be important in the proof of Theorem~\ref{thm_monoMeasRel}.

\begin{lemma}\label{lem_pinchingVsRotPetz}
For any $\sigma \in \Pos(A)$, any $\cN \in \TPCP(A,B)$, and any $n\in \N$, let $\cR_{\sigma,\cN}^n$ be the pinching recovery map defined in~\eqref{eq_recMap} and let $\cT^{\varphi,\vartheta}_{\sigma,\cN}$ be the rotated Petz recovery map defined in~\eqref{eq_rotPetzMap}. We have
\begin{align}
\cR_{\sigma,\cN}^n\big(\cdot\big)  = \frac{1}{(2 \pi)^{d_1}} \int_{[0,2\pi]^{\times d_1}}\di \vartheta  \frac{1}{(2 \pi)^{d_2}}\int_{[0,2\pi]^{\times d_2}}  \di \varphi  (\cT^{\varphi,\vartheta}_{\sigma,\cN})^{\otimes n}\big(\cdot\big)  \ .
\end{align}
\end{lemma}

\begin{proof}
Recall the definition $U_{\cN(\sigma)}^{\vartheta} = \sum_{k \in [d_2]} \exp(\ci \vartheta_k) P_k$.
For each $\mu = (\mu_1, \ldots, \mu_{d_2})$ with $\mu_k \in \mathbb{N}$ and $\sum_{i =1}^{d_2} \mu_{i} = n$, define the corresponding subspaces $P_{\mu}^n := \sum_{x^n} \bigotimes_{m=1}^n P_{x_m}$, where the sum goes over all sequences $x^n = (x_1, \ldots, x_n)$ of type $\mu$.
Using this, we can write 
  \begin{align}
    \bigl( U_{\mathcal{N}(\sigma)}^{\varphi} \bigr)^{\otimes n} 
= \Bigl( \sum_{k \in [d_2]} \exp(\ci \varphi_k) P_k \Bigr)^{\otimes n} 
= \sum_{\mu} \exp\Bigl( \ci \sum_{k \in [d_2]} \varphi_k  \mu_k \Bigr) P_{\mu}^n 
= \sum_{\mu} \exp\big( \ci \, \langle \mu , \varphi \rangle \big) P_{\mu}^n \,,
  \end{align}
  where we used the inner product $\langle \mu , \varphi \rangle := \sum_{k \in [d_2]} \varphi_k  \mu_k$.
  Let us next take a closer look at the inner part of the map $\big(\cT_{\sigma,\cN}^{\varphi,\vartheta}\big)^{\otimes n}$ in~\eqref{eq_rotPetzMap}, i.e., the expression
  \begin{align}
\big( U_{\mathcal{N}(\sigma)}^{\varphi} \big)^{\otimes n} X_{B^n} \big( U_{\mathcal{N}(\sigma)}^{\varphi\dag} \big)^{\otimes n} 
= \sum_{\mu,\mu'} \exp\big( i \langle \mu -\mu' , \varphi \rangle \big) P_{\mu}^n X_{B^n} P_{\mu'}^n \ .
  \end{align}
  This is the only part of~\eqref{eq_rotPetzMap} that depends on $\varphi$. If we integrate $\varphi$ over $[0,2\pi]^{\times d_2}$, we find that
  \begin{align}
  \frac{1}{(2 \pi)^{d_2}}   \int_{[0,2\pi]^{\times d_2}} \di \varphi\ \exp\big( i \langle \mu -\mu' , \varphi \rangle \big) = \indic_{\{ \mu = \mu' \}} \,,
  \end{align}
  and, thus,
  \begin{align}
   \frac{1}{(2 \pi)^{d_2}}  \int_{[0,2\pi]^{\times d_2}} \di \varphi\  \big( U_{\mathcal{N}(\sigma)}^{\varphi} \big)^{\otimes n} X_{B^n} \big( U_{\mathcal{N}(\sigma)}^{\varphi\dag} \big)^{\otimes n}
   = \sum_{\mu} P_{\mu}^n X_{B^n} P_{\mu}^n 
   = \cP_{\cN(\sigma),n} (X_{B^n}) \,.
  \end{align}
  
  The analogous argument applied to the unitary rotations $\big(U_{\sigma}^{\vartheta}\big)^{\otimes n}$ integrated over $\vartheta$ reveals the pinching map $\cP_{\sigma,n} (\cdot)$. Combining these arguments concludes the proof of the lemma.
\end{proof}

 
 \section{Main Result and Proof} \label{sec_results}
 
 Our first result gives a lower bound on the difference between the relative entropy before and after applying a quantum channel. The lower bound is given in terms of the limit of a sequence of relative entropies between $\rho^{\otimes n}$ and the state recovered from $\cN(\rho)^{\otimes n}$ using the pinching recovery map.
 
 \begin{proposition} \label{prop_main}
 For any $\sigma \in \Pos(A)$ and any trace-preserving tensor-stable positive linear map $\cN\in \LL(A,B)$ the pinching recovery map $\cR^n_{\sigma,\cN}$ defined in~\eqref{eq_recMap} satisfies for any $\rho \in \St_{\sigma}(A)$
 \begin{align}
 &D(\rho \, \| \sigma) - D\bigl(\cN(\rho) \big\| \cN(\sigma) \bigr)
  \geq \liminf_{n \to \infty} \frac{1}{n} D\bigl(\rho^{\otimes n} \big\| (\cR^n_{\sigma,\cN} \circ \cN^{\otimes n})(\rho^{\otimes n})\bigr) \ .
 \end{align}
\end{proposition}
\begin{remark}
Proposition~\ref{prop_main} offers a proof for the monotonicity of the relative entropy based on the concavity and monotonicity of the operator logarithm, the operator Jensen inequality, and the non-negativity of the relative entropy (i.e., Klein's inequality).
\end{remark}

We note that combining Proposition~\ref{prop_main} with a recent result showing that the \emph{fidelity of recovery} is multiplicative~\cite{TB15} directly reproduces the second inequality of~\eqref{eq_FR}. (The fidelity of recovery for a tripartite state $\rho_{ABC} \in \St(A \otimes B \otimes C)$ is defined as $F(A;C|B)_{\rho}:= \max_{\cR \in \TPCP(B,BC)} F(\rho_{ABC},\cR(\rho_{AB}))$ has been introduced in~\cite{SW14} and its properties were studied there and further in~\cite{TB15}.) A stronger result can be obtained from Proposition~\ref{prop_main} by further employing the structure of the pinching recovery map.

\begin{theorem} \label{thm_monoMeasRel}
For any $\sigma \in \Pos(A)$, any $\rho \in \St_{\sigma}(A)$, and any trace-preserving tensor-stable positive linear map $\cN \in \LL(A,B)$ there exists a recovery map $\cR_{\sigma,\cN,\rho} \in \TT_{\sigma,\cN}$ with $\TT_{\sigma,\cN}$ given in~\eqref{eq_convexHullRotatedPetz}, such that
\begin{align} \label{eq_generalizedMarkBound}
&D(\rho \| \sigma) - D\bigl(\cN(\rho) \big\| \cN(\sigma) \bigr)  
\geq \MD\bigl(\rho \, \big\| (\cR_{\sigma,\cN,\rho} \circ \cN)(\rho)\bigr) 
\geq -2 \log  F\bigl(\rho , (\cR_{\sigma,\cN,\rho} \circ \cN)(\rho)\bigr) \ .
\end{align}
\end{theorem}

\begin{remark}
The recovery map stated in Theorem~\ref{thm_monoMeasRel} has the property that it perfectly recovers $\sigma$ from $\cN(\sigma)$, i.e., $(\cR_{\sigma,\cN,\rho} \circ \cN)(\sigma) = \sigma$. This follows directly by the fact that the unitaries $U_{\sigma}^{\vartheta}$ and $U_{\cN(\sigma)}^{\varphi}$ (defined in Section~\ref{sec_variantsRecMaps}) commute with $\sigma$ and $\cN(\sigma)$, respectively, for any $\vartheta$ and any $\varphi$.
 We further note that the recovery map predicted by Theorem~\ref{thm_monoMeasRel} has a potential dependence on $\rho$ as the weight of the convex sum in the definition of $\TT_{\sigma,\cN}$ (see Equation~\eqref{eq_convexHullRotatedPetz}) could depend on $\rho$. In other words, the recovery map predicted by Theorem~\ref{thm_monoMeasRel} does not possess the universality property that has been established in~\cite{SFR15} for the conditional mutual information lower bound.
\end{remark}

\begin{remark}
Our new bound in terms of measured relative entropy in~\eqref{eq_generalizedMarkBound} is always strictly stronger than the bound in terms of the fidelity, except for the case of perfect recoverability where both bounds vanish.\footnote{To verify this, note that $F(p, q) = F(\rho , \sigma)$ for some measurement and post-measurement distributions $p$ and $q$ corresponding to $\rho$ and $\sigma$, respectively. Hence, $F(\rho,\sigma) \neq 1$ implies that $p \neq q$, and the strict monotonicity of the (commutative) R\'enyi divergence ensures that $D_{\mathbb{M}}(\rho\|\sigma) \geq D(p\|q) > D_{\nicefrac12}(p\|q) = -\log F(p,q)$.}
Moreover, as discussed in~\cite{BHOS14}, the gap between the two bounds can be of order $\log d$ where $d$ is the dimension of the quantum system. Finally, in the commutative case our bound coincides with the best known classical result.
\end{remark}

\begin{remark}We further note that Theorem~\ref{thm_monoMeasRel} holds for any trace-preserving tensor-stable positive linear map $\cN\in \LL(A,B)$ which is more general than only allowing for trace-preserving completely positive maps~\cite{MRW15}.
\end{remark}

By choosing $\rho = \rho_{ABC}$, $\sigma=\id_A \otimes \rho_{BC}$, and $\cN(\cdot)=\tr_C(\cdot)$, Theorem~\ref{thm_monoMeasRel} immediately reproduces the lower bound for the conditional mutual information given in~\eqref{eq_FR}. In addition, we obtain additional information about the structure of the recovery map such as that it maps $\rho_B$ to $\rho_{BC}$. 
\begin{corollary}
For any $\rho_{ABC} \in \St(A \otimes B \otimes C)$ there exists a recovery map $\cR_{B \to BC} \in \TT_{\rho_{BC},\tr_C}$ with $\TT_{\rho_{BC},\tr_C}$ given in~\eqref{eq_convexHullRotatedPetz}, such that 
\begin{align}
I(A:C|B)_{\rho} \geq \MD\bigl(\rho_{ABC} \big\|  \cR_{B \to BC}(\rho_{AB}) \bigr)
\geq - 2 \log F\bigl(\rho_{ABC},\cR_{B \to BC}(\rho_{AB}) \bigr) \ .
\end{align}
\end{corollary}

Whereas Theorem~\ref{thm_monoMeasRel} provides an lower bound for the relative entropy difference in terms of a distance to a recovered state it is of interest if there also exists an upper bound for the relative entropy in terms of recoverability. The following proposition proves such an upper bound in terms of the \emph{max relative entropy} that is defined as $D_{\max}(\rho \| \sigma):= \inf\{\gamma : \rho \leq 2^{\gamma} \sigma \}$ for two density operators $\rho$ and $\sigma$. 
\begin{proposition} \label{prop_UB}
For any $\sigma \in \Pos(A)$, any $\rho \in \St_{\sigma}(A)$, and any trace-preserving tensor-stable positive linear map $\cN \in \LL(A,B)$ we have
\begin{align}
D(\rho \| \sigma) - D\bigl(\cN(\rho) \big\| \cN(\sigma) \bigr) 
\leq \min_{\cR \in \mathrm{TPTSP}(B,A)}\{D_{\max}\bigl((\cR \circ \cN)(\sigma)\| \sigma \bigr)  :   (\cR \circ \cN)(\rho)=\rho \} \ ,
\end{align}
where $ \mathrm{TPTSP}(B,A)$ denotes the set of trace-preserving tensor-stable positive linear maps from $A$ to $B$.
\end{proposition}

\begin{remark}
A conceptually different upper bound for the relative entropy difference has been established in~\cite{Wilde15}. We note that the bound in~\cite{Wilde15} is expressed in terms of the max relative entropy as well but involves a maximization over recovery maps of a particular form. In contrast, our upper bound is expressed in terms of a minimization over recovery maps and thus can be further bounded by any choice of such a recovery map.
\end{remark}

\subsection{Proof of Proposition~\ref{prop_main}}

The proof of Proposition~\ref{prop_main} uses basic properties of pinching maps (see Section~\ref{sec_variantsRecMaps}) and the operator monotonicity and concavity of the logarithm.
We start with a preparatory lemma that will be used in the proof of Proposition~\ref{prop_main}. It is a simple consequence of the operator concavity of the logarithm.

\begin{lemma} \label{lem_opJensen}
For any $n \in \N$, any $\rho^n \in \Pos(A^{\otimes n})$, any $\sigma^n \in \Pos(B^{\otimes n})$, any trace-preserving $n$-tensor-stable positive linear map $\cN \in \LL(A,B)$, we have
\begin{align}
\tr\bigl(\cN^{\otimes n}(\rho^n) \log \sigma^n \bigr)  \leq \tr\bigl( \rho^n \log \, \cN^{\dagger\,\otimes n}(\sigma^n) \bigr) \ .
\end{align}
\end{lemma}
\begin{proof}
Since $\cN$ is trace-preserving and $n$-tensor-stable positive this implies that $\cN^\dagger$ is unital and $n$-tensor-stable positive. Thus by the operator Jensen inequality for positive unital maps proven in~\cite[Theorem~2.1]{choi74} (see also~\cite{davis57,hansen03} and~\cite[Footnote~1]{hansen07}) and the concavity of the operator logarithm we obtain for any $n \in \N$
\begin{align}
\tr\bigl(\cN^{\otimes n}(\rho^n) \log \sigma^n \bigr)
=\tr\bigl(\rho^n \cN^{\dagger \, \otimes n}( \log \sigma^n) \bigr) 
\leq \tr\bigl(\rho^n \log \cN^{\dagger \, \otimes n}(\sigma^n) \bigr)  \ .
\end{align}
\end{proof}

\begin{proof}[Proof of Proposition~\ref{prop_main}]
 First, note that $\cP_{H}(X)$ commutes with $H$ for any $X \in \Pos(A)$, and thus
 \begin{align*}
 \log \bigl( H^{\frac12} \cP_{H}(X) H^{\frac12} \bigr) = \log \cP_{H}(X) + \log H \ .
 \end{align*}
 Using this fact, we find for any $n \in \N$ that
\begin{align} \label{eq_firstPart}
 &D\bigl(\rho^{\otimes n} \big\| (\cR^n_{\sigma,\cN} \circ \cN^{\otimes n})(\rho^{\otimes n})\bigr) \nonumber \\
 &\hspace{4mm}=\! -n H(\rho)\! -\! \tr\Big(\!\rho^{\otimes n} \! \log\! \Big[\!  (\sigma^{\frac{1}{2}})^{\otimes n} \cP_{\sigma,n}\big(\!(\cN^{\dagger})^{\otimes n}\! \Bigl[\! \bigl(\cN(\sigma)^{-\frac{1}{2}}\bigr)^{\!\otimes n}  \cP_{\cN(\sigma),n} (\cN(\rho)^{\otimes n}) \bigl(\cN(\sigma)^{-\frac{1}{2}}\bigr)^{\!\otimes n}  \Bigr]  \big)  (\sigma^{\frac{1}{2}})^{\otimes n} \Big] \Big) \nonumber \\
 &\hspace{4mm}= n D(\rho \| \sigma) - \tr\Big(\rho^{\otimes n} \log \cP_{\sigma,n}\Big((\cN^{\dagger})^{\otimes n} \Big[ \big(\cN(\sigma)^{-\frac{1}{2}}\big)^{\otimes n} \cP_{\cN(\sigma),n} (\cN(\rho)^{\otimes n}) \bigl(\cN(\sigma)^{-\frac{1}{2}}\bigr)^{\otimes n}  \Big]  \Big)\Big) \nonumber \\
 &\hspace{4mm}\leq n D(\rho \| \sigma) - \tr\Big(\rho^{\otimes n} \log (\cN^{\dagger})^{\otimes n} \Big[ \big(\cN(\sigma)^{-\frac{1}{2}}\big)^{\otimes n} \cP_{\cN(\sigma),n} (\cN(\rho)^{\otimes n}) \big(\cN(\sigma)^{-\frac{1}{2}}\big)^{\otimes n}  \Big]  \Big) + O(\log n) \ .
\end{align}
In the last step we used the pinching inequality together with~\eqref{eq_csiszar} and the fact $\tr(X \log Y) \geq \tr(X \log Z)$ for $X\geq 0$ and $Y \geq Z$ which is due to the monotonicity of the operator logarithm~\cite[Exercise~4.2.5]{bhatia_psd_book}.
Next, Lemma~\ref{lem_opJensen} ensures that we have
\begin{align} \label{eq_secPart}
 &\tr\Big(\rho^{\otimes n} \log (\cN^{\dagger})^{\otimes n} \Big[ \big(\cN(\sigma)^{-\frac{1}{2}}\big)^{\otimes n} \cP_{\cN(\sigma),n} (\cN(\rho)^{\otimes n})  \big(\cN(\sigma)^{-\frac{1}{2}}\big)^{\otimes n}  \Bigr]  \Big) \nonumber  \\
  & \hspace{40mm}\geq  \tr\Big( \cN(\rho)^{\otimes n} \log \Big[  \big(\cN(\sigma)^{-\frac{1}{2}}\big)^{\otimes n} \cP_{\cN(\sigma),n} (\cN(\rho)^{\otimes n})\big(\cN(\sigma)^{-\frac{1}{2}}\big)^{\otimes n} \Big]  \Big) \ .
\end{align}
Combining~\eqref{eq_firstPart} and~\eqref{eq_secPart}, we find by applying the same arguments as in~\eqref{eq_firstPart}
\begin{align}
&D\bigl(\rho^{\otimes n} \big\| (\cR^n_{\sigma,\cN} \circ \cN^{\otimes n})(\rho^{\otimes n})\bigr) \nonumber \\
& \hspace{15mm}\leq n D(\rho \| \sigma) + n \, \tr\bigl( \cN(\rho) \log \cN(\sigma) \bigr)  - \tr\Bigl(\cN(\rho)^{\otimes n} \log \cP_{\cN(\sigma),n} (\cN(\rho)^{\otimes n})  \Bigr) + O(\log n)  \\
&\hspace{15mm}\leq n D(\rho \, \| \sigma) - n D\bigl(\cN(\rho) \big\| \cN(\sigma) \bigr) + O(\log n) \ ,
\end{align}
which proves the statement of Proposition~\ref{prop_main}.
\end{proof}

\subsection{Proof of Theorem~\ref{thm_monoMeasRel}}

We start by proving a basic property of the measured relative entropy as defined in~\eqref{eq_defMeasRelEnt}. An alternative proof for Theorem~\ref{thm_monoMeasRel} can be found in Appendix~\ref{ap_piani}.
As shown in~\cite[Sections~11.9 and 11.10]{petz_Entropybook} (see also~\cite{Petz_variational88,FBT15}), there exists a variational characterization for the measured relative entropy of the form
\begin{align} \label{eq_variationalMeasRelEnt}
\MD(\rho \| \sigma) &= \sup_{\omega >0} \tr(\rho \log \omega) - \log \tr(\sigma \omega) \\
&= \sup_{\omega >0} \tr(\rho \log \omega) + 1 - \tr(\sigma \omega) \ .
\end{align}

\begin{lemma} \label{lem_singleLetterMeasRelEnt}
Let $\X$ be a compact space. For any probability measure $\mu \in \bPr(\X)$, any family $\{\sigma_x\}_{x \in \X}$ such that $\sigma_x \in \Pos(A)$ for all $x \in \X$, $\sigma=\int \mu(\di x) \sigma_x$, any $\rho \in \St_{\sigma}(A)$ and any $n \in \N$, we have
\begin{align}
\frac{1}{n}\MD\Bigl( \rho^{\otimes n} \Big\| \int \mu(\di x) \sigma_x^{\otimes n}  \Bigr) \geq \min_{\sigma \in \conv\{\sigma_x : x \in \X \} } \MD(\rho \| \sigma) \ .
\end{align}
\end{lemma}
\begin{proof}
With the variational characterization for the measured relative entropy given in~\eqref{eq_variationalMeasRelEnt} we find
\begin{align}
\MD\Bigl( \rho^{\otimes n} \Big\| \int \mu(\di x) \sigma_x^{\otimes n}  \Bigr)  
&\geq \sup_{\omega > 0} \tr \bigl( \rho^{\otimes n} \log \omega^{\otimes n} \bigr) - \log \tr \Bigl( \int \mu(\di x) \sigma_x^{\otimes n} \omega^{\otimes n} \Bigr)  \\
&\geq \sup_{\omega >0} \min_{x \in \X} n \tr(\rho \log \omega) - n \log\tr(\sigma_x \omega) \ .
\end{align}
For $x \in \R_+$, clearly $\log x \leq x -1$ and thus $- \log \tr(\sigma \omega) \geq 1 - \tr(\sigma \omega)$ for all $\omega >0$. This implies that
\begin{align}
\MD\Bigl( \rho^{\otimes n} \Big\| \int \mu(\di x) \sigma_x^{\otimes n}  \Bigr) 
&\geq n \sup_{\omega >0} \min_{x \in \X} \tr(\rho \log \omega) + 1 - \tr(\sigma_x \omega)  \\
&\geq n \sup_{\omega >0} \min_{\sigma \in \conv\{\sigma_x : x \in \X \}} \tr(\rho \log \omega) + 1 - \tr(\sigma \omega) \ .
\end{align}

The function $\omega \mapsto \tr(\rho \log \omega) + 1 - \tr(\sigma \omega) $ is clearly concave and the function $\sigma \mapsto  \tr(\rho \log \omega) + 1 - \tr(\sigma \omega) $ is linear. The set $\conv\{\sigma_x : x \in \X \}$ is compact and convex and the set of strictly positive operators is convex.
 As a result we can apply Sion's minimax theorem~\cite{Sion58} which gives
\begin{align}
\frac{1}{n}\MD\Bigl( \rho^{\otimes n} \Big\| \int \mu(\di x) \sigma_x^{\otimes n}  \Bigr) 
&\geq \min_{\sigma \in \conv\{\sigma_x : x \in \X \}} \sup_{\omega >0}   \tr(\rho \log \omega) + 1 - \tr(\sigma \omega)  \\
& = \min_{\sigma \in \conv\{\sigma_x : x \in \X \} } \MD(\rho \| \sigma)  \ ,
\end{align}
where the final step follows by the variational characterization of the measured relative entropy given in~\eqref{eq_variationalMeasRelEnt}.
\end{proof}

For any $\sigma \in \Pos(A)$ and $\cN \in \TPCP(A,B)$ let $\TT_{\sigma,\cN}$ denote the convex hull of all rotated Petz recovery maps (defined in Equation~\eqref{eq_convexHullRotatedPetz}) and let $\bPr(\cQ)$ be the set of probability measures on $\cQ:=[0,2\pi]^{\times d_1}\times[0,2\pi]^{\times d_2}$.
Combining this with~Proposition~\ref{prop_main}, the monotonicity of the relative entropy under trace-preserving completely positive maps and Lemma~\ref{lem_pinchingVsRotPetz} gives
\begin{align}
&D(\rho \| \sigma) - D\bigl(\cN(\rho) \big\| \cN(\sigma) \bigr) \nonumber \\
&\hspace{20mm}\geq \liminf_{n \to \infty} \frac{1}{n} D\Bigl(\rho^{\otimes n} \big\| \cR^n_{\sigma,\cN}\bigl( \cN(\rho)^{\otimes n} \bigr) \Bigr)  \\
& \hspace{20mm}\geq \liminf_{n \to \infty} \frac{1}{n} \MD\Bigl(\rho^{\otimes n} \big\| \cR^n_{\sigma,\cN}\bigl( \cN(\rho)^{\otimes n} \bigr) \Bigr)  \\
&\hspace{20mm}= \liminf_{n \to \infty} \frac{1}{n} \MD\Big(\rho^{\otimes n} \Big\| \frac{1}{(2 \pi)^{d_1}}  \int_{[0,2\pi]^{\times d_1}}  \di \vartheta  \, \,\frac{1}{(2 \pi)^{d_2}} \int_{[0,2\pi]^{\times d_2}} \di \varphi  (\cT^{\varphi,\vartheta}_{\sigma,\cN} \circ \cN)(\rho)^{\otimes n} \Big)  \\
&\hspace{20mm}\geq \min_{\mu \in \bPr(\cQ)} \MD\Bigl(\rho \, \Big\| \int  \di \mu(\varphi,\vartheta) \, (\cT_{\sigma,\cN}^{\varphi,\vartheta} \circ \cN)(\rho)\Bigr) \ ,
 \end{align}
where the final inequality follows by Lemma~\ref{lem_singleLetterMeasRelEnt}. \qed

\subsection{Proof of Proposition~\ref{prop_UB}}
\begin{proof}
Let $\cR \in \TPCP(B,A)$ be such that $(\cR \circ \cN)(\rho)=\rho$ and define $\lambda:= D_{\max}((\cR \circ \cN)(\sigma)\| \sigma)$. By definition of the max-relative entropy we have
\begin{align}\label{eq_byDef}
(\cR \circ \cN)(\sigma) \leq 2^\lambda \sigma \ .
\end{align}
Combining this with the monotonicity of the operator logarithm gives
\begin{align}
D(\rho \| \sigma) - D\bigl( \cN(\rho) \| \cN(\sigma) \bigr)
&= D\bigl((\cR \circ \cN)(\rho) \| \sigma \bigr) - D\bigl( \cN(\rho) \| \cN(\sigma) \bigr)  \\
&\leq \lambda + D\bigl((\cR \circ \cN)(\rho) \| (\cR \circ \cN)(\sigma)  \bigr) - D\bigl( \cN(\rho) \| \cN(\sigma) \bigr)  \\
&\leq \lambda \ ,
\end{align}
where the final inequality uses the data processing inequality. Since this argument applies to every recovery map $\cR \in \TPCP(B,A)$ that satisfies $(\cR \circ \cN)(\rho)=\rho$ the assertion follows.
\end{proof}

\section{Discussion}
The recovery map predicted by Theorem~\ref{thm_monoMeasRel} has a potential dependence on the state $\rho$. It would be natural to assume that there exists a recovery map $\cR \in \TT_{\sigma,\cN}$ satisfying~\eqref{eq_generalizedMarkBound} that is universal in the sense that it only depends on $\sigma$ and $\cN$. This has been proven recently~\cite{SFR15} for the case of the conditional mutual information (i.e., for $\rho=\rho_{ABC}$, $\sigma=\rho_{BC}$, and $\cN(\cdot)=\tr_C(\cdot)$). As it turns out, Theorem~\ref{thm_monoMeasRel} together with the techniques used in~\cite{SFR15} implies a universal version of Theorem~\ref{thm_monoMeasRel}~\cite{JRSWW15}. 



\appendix

\section*{Appendices}

\section{Alternative proof of Theorem~\ref{thm_monoMeasRel}} \label{ap_piani}
This appendix presents an alternative proof of Theorem~\ref{thm_monoMeasRel}.
For any $\sigma \in \Pos(A)$ and $\cN \in \TPCP(A,B)$ let $\TT_{\sigma,\cN}$ denote the convex hull of all rotated Petz recovery maps (defined in Equation~\eqref{eq_convexHullRotatedPetz}) and let $\bPr(\cQ)$ be the set of probability measures on $\cQ:=[0,2\pi]^{\times d_1}\times[0,2\pi]^{\times d_2}$. According to Theorem~1 of~\cite{piani09} (see also Lemma~8 and the proof of Proposition~4 in~\cite{BHOS14})\footnote{To see how Theorem~1 of~\cite{piani09} can be applied it may be helpful to consider the reference set $\mathrm{K}:=\cup_{k \in \N} \conv\{ {(\cT_{\sigma,\cN}\circ \cN)}(\rho)^{\otimes k} : \cT_{\sigma,\cN} \in \TT_{\sigma,\cN} \}$.} 
we find for any $n \in \N$
\begin{align} \label{eq_argument}
&\min_{\mu \in \bPr(\cQ)} D\Bigl(\rho^{\otimes n} \Big\| \int \di \mu(\varphi,\vartheta) \, (\cT_{\sigma,\cN}^{\varphi,\vartheta} \circ \cN)(\rho)^{\otimes n} \Bigr) \nonumber  \\
&\hspace{0mm}\geq \! \min_{\mu \in \bPr(\cQ)} \! \MD\Bigl(\rho \, \Big\| \int \di \mu(\varphi,\vartheta) \, (\cT_{\sigma,\cN}^{\varphi,\vartheta} \circ \cN)(\rho)\Bigr) + \min_{\mu \in \bPr(\cQ)}  D\Bigl(\rho^{\otimes n-1} \Big\| \int \di\mu(\varphi,\vartheta)  \, (\cT_{\sigma,\cN}^{\varphi,\vartheta} \circ \cN)(\rho)^{\otimes n-1} \Bigr) \ .
\end{align}
We note that the minima are attained. To see this, we first remark that the set $\cQ$ is clearly compact and as a result, the set $\bPr(\cQ)$ is weak* compact~\cite[Theorem~15.11]{aliprantis07}. Furthermore, the function $\bPr(\cQ) \ni \mu \mapsto D\bigl(\rho^{\otimes n} \| \int \di \mu(\varphi,\vartheta) \, {(\cT^{\varphi,\vartheta}_{\sigma,\cN} \circ \cN)(\rho)^{\otimes n}} \bigr)\in  \R_+$ is lower semicontinuous with respect to the weak* topology, since expectation values are continuous in the weak* topology and the relative entropy is lower semicontinuous~\cite[Exercise~7.22 and Theorem~11.6]{holevo_book}.
The extreme value theorem thus ensures that the minima are attained.
Iterating the argument presented in~\eqref{eq_argument} $n$ times shows that for any $n \in \N$
\begin{align}
&\frac{1}{n} \min_{\mu \in \bPr(\cQ)} D\Bigl(\rho^{\otimes n} \Big\| \int \di\mu(\varphi,\vartheta) \, (\cT_{\sigma,\cN}^{\varphi,\vartheta} \circ \cN)(\rho)^{\otimes n} \Bigr) \geq \min_{\mu \in \bPr(\cQ)} \MD\Bigl(\rho \, \Big\| \int \di \mu(\varphi,\vartheta) \, (\cT_{\sigma,\cN}^{\varphi,\vartheta} \circ \cN)(\rho)\Bigr) \ .
\end{align}

Combining this with~Proposition~\ref{prop_main} and Lemma~\ref{lem_pinchingVsRotPetz} gives
\begin{align}
&D(\rho \| \sigma) - D\bigl(\cN(\rho) \big\| \cN(\sigma) \bigr)  \nonumber \\
&\hspace{20mm}\geq \liminf_{n \to \infty} \frac{1}{n} D\Bigl(\rho^{\otimes n} \big\| \cR^n_{\sigma,\cN}\bigl( \cN(\rho)^{\otimes n} \bigr) \Bigr)  \\
&\hspace{20mm}= \liminf_{n \to \infty} \frac{1}{n} D\Big(\rho^{\otimes n} \Big\| \frac{1}{(2 \pi)^{d_1}}  \int_{[0,2\pi]^{\times d_1}}  \di \vartheta  \, \,\frac{1}{(2 \pi)^{d_2}} \int_{[0,2\pi]^{\times d_2}} \di \varphi  (\cT^{\varphi,\vartheta}_{\sigma,\cN} \circ \cN)(\rho)^{\otimes n} \Big)  \\
&\hspace{20mm}\geq \liminf_{n \to \infty} \frac{1}{n}\! \min_{\mu \in \bPr(\cQ)}\!\! D\Bigl(\rho^{\otimes n} \Big\| \int \di \mu(\varphi,\vartheta) \, (\cT_{\sigma,\cN}^{\varphi,\vartheta} \circ \cN)(\rho)^{\otimes n} \Bigr)  \\
&\hspace{20mm}\geq \min_{\mu \in \bPr(\cQ)} \MD\Bigl(\rho \, \Big\| \int  \di \mu(\varphi,\vartheta) \, (\cT_{\sigma,\cN}^{\varphi,\vartheta} \circ \cN)(\rho)\Bigr) \ ,
 \end{align}
which proves the first inequality of Theorem~\ref{thm_monoMeasRel}. The second inequality is a consequence of the monotonicity of quantum R\'enyi divergence in the order parameter~\cite{MLDSFT13} and of the fact that for any two states there exists an optimal measurement that does not increase their fidelity~\cite[Section~3.3]{Fuc96}.
 \qed

\section*{Acknowledgments}
We thank Renato Renner and Mark M.\ Wilde for helpful discussions. DS acknowledges support by the European Research Council (ERC) via grant No.~258932, by the Swiss National Science Foundation (SNSF) via the National Centre of Competence in Research ``QSIT'', and by the European Commission via the project ``RAQUEL''. 
MT is funded by an University of Sydney Postdoctoral Fellowship and acknowledges support from the ARC Centre of Excellence for Engineered Quantum Systems (EQUS). 
AWH was funded by NSF grants CCF-1111382 and CCF-1452616 and ARO contract W911NF-12-1-0486.
  \bibliographystyle{abbrv}
  
  \bibliography{bibliofile}

\end{document}